\documentclass[11pt,letterpaper]{article}

\usepackage{authblk}

\usepackage[margin=1in]{geometry}

\usepackage{mathpazo}

\usepackage{enumerate}
\usepackage{enumitem}

\usepackage{microtype}

\usepackage{tikz}
\usetikzlibrary{quantikz2}

\usepackage{ytableau}

\usepackage{float}
\usepackage[font=small]{caption}

\usepackage[colorlinks = true]{hyperref}
\definecolor{darkred}  {rgb}{0.5,0,0}
\definecolor{darkblue} {rgb}{0,0,0.5}
\definecolor{darkgreen}{rgb}{0,0.5,0}
\hypersetup{
  urlcolor   = blue,         
  linkcolor  = darkblue,     
  citecolor  = darkgreen,    
  filecolor  = darkred       
}

\usepackage{amsmath,amssymb,amsfonts,amsthm,amstext}

\usepackage{bm}

\usepackage{etoolbox}

\usepackage{mathtools}
\mathtoolsset{centercolon}
\makeatletter
\protected\def\tikz@nonactivecolon{\ifmmode\mathrel{\mathop\ordinarycolon}\else:\fi}
\makeatother

\usepackage{tcolorbox}

\usepackage{algorithm}
\usepackage[beginComment=]{algpseudocodex}

\usepackage[nameinlink]{cleveref}
\crefname{lemma}{Lemma}{Lemmas}
\crefname{proposition}{Proposition}{Propositions}
\crefname{definition}{Definition}{Definitions}
\crefname{theorem}{Theorem}{Theorems}
\crefname{corollary}{Corollary}{Corollaries}
\crefname{claim}{Claim}{Claims}
\crefname{section}{Section}{Sections}
\crefname{appendix}{Appendix}{Appendices}
\crefname{figure}{Fig.}{Figs.}
\crefname{table}{Table}{Tables}
\crefname{algorithm}{Algorithm}{Algorithms}
\crefname{enumi}{part}{parts}

\usepackage[retainorgcmds]{IEEEtrantools}

\usepackage{tocloft}

\newcommand{\x}{\otimes}

\renewcommand{\Pr}{\mathrm{Pr}}

\DeclarePairedDelimiter{\set}{\lbrace}{\rbrace}

\DeclarePairedDelimiter{\norm}{\lVert}{\rVert}
\DeclarePairedDelimiter{\of}{\lparen}{\rparen}

\DeclarePairedDelimiter{\floor}{\lfloor}{\rfloor}
\DeclarePairedDelimiter{\ceiling}{\lceil}{\rceil}

\let\originalleft\left
\let\originalright\right
\renewcommand{\left}{\mathopen{}\mathclose\bgroup\originalleft}
\renewcommand{\right}{\aftergroup\egroup\originalright}

\usepackage{trimclip}
\usepackage{opencolor}

\makeatletter 
\newcommand*\unevendelim[3]{{\mathpalette\unevendelim@{{#1}{#2}{#3}}}}
\def\unevendelim@#1#2{\unevendelim@@{#1}#2}%
\def\unevendelim@@#1#2#3#4{%
    \sbox0{$\m@th#1#4$}%
    \sbox6{$\m@th#1\{\}$}%
    \unevendelim@@@{#1}{\left#2}{\right.}%
    \copy0
    \unevendelim@@@{#1}{\left.}{\right#3}%
}
\def\unevendelim@@@#1#2#3{%
  \sbox2{$\m@th#1#2\rule{0pt}{\ht0}#3$}%
  \sbox4{$\m@th#1#2\rule[-\dp0]{0pt}{\dp0}#3$}%
  \ifdim\ht2>\ht4
    \ooalign{\clipbox{0pt {\dimexpr\dp2-\dp6\relax} 0pt 0pt}{\copy2}\cr
             \raisebox{-\dp6}{\clipbox{0pt 0pt 0pt {\dimexpr2\ht2-\ht4+\dp6}}{\copy2}}\cr}
  \else
    \ooalign{\raisebox{\ht6}{\clipbox{0pt {\dimexpr2\dp4-\dp2+\ht6} 0pt 0pt}{\copy4}}\cr
             \clipbox{0pt 0pt 0pt {\dimexpr\ht4-\ht6}}{\copy4}\cr}
  \fi
}
\makeatother  

\DeclareMathOperator{\Tr}{Tr}

\newcommand{\SWAPGADGET}{\normalfont\textsc{Swapgadget}}
\newcommand{\SWAP}{\normalfont\textsc{Swap}}
\newcommand{\PURIFY}{\normalfont\textsc{Purify}}

\newcommand{\C}{\mathbb{C}}

\newcommand{\vlambda}{\bm{\lambda}} 

\newcommand{\e}{\delta}

\usepackage{color}

\newtheorem{theorem}{Theorem}
\newtheorem{lemma}[theorem]{Lemma}
\newtheorem{proposition}[theorem]{Proposition}
\newtheorem{definition}[theorem]{Definition}
\newtheorem{corollary}[theorem]{Corollary}

\newcommand{\bp}{\textbf{q}}
\newcommand{\blambda}{\boldsymbol{\lambda}}
\newcommand{\eff}{\text{eff}}
\newcommand{\nlog}{\ln}
\usepackage{calligra}
\usepackage[T1]{fontenc}
\newcommand{\st}{{\calligra s}}

\begin{document}

\title{Streaming quantum state purification for general mixed states}

\author[1]{Daniel Grier}
\author[2,3]{Debbie Leung}
\author[4,3]{Zhi Li}
\author[5,6]{Hakop Pashayan}
\author[2,7]{Luke Schaeffer}

\renewcommand*{\Authfont}{\normalsize}
\renewcommand\Affilfont{\itshape\footnotesize}

\affil[1]{Department of Mathematics and Department of Computer Science and Engineering, University of California, San Diego}

\affil[2]{Institute for Quantum Computing, University of Waterloo}

\affil[3]{Perimeter Institute for Theoretical Physics}

\affil[4]{National Research Council of Canada}

\affil[5]{Dahlem Center for Complex Quantum Systems, Freie Universität Berlin}

\affil[6]{Hon Hai (Foxconn) Research Institute}

\affil[7]{Department of Computer Science, Institute for Advanced Computer Studies, and Joint Center for Quantum \break Information and Computer Science, University of Maryland}

\date{}

\maketitle

\begin{abstract}
Given multiple copies of a mixed quantum state with an unknown, nondegenerate principal eigenspace, quantum state purification is the task of recovering a quantum state that is closer to the principal eigenstate.
A streaming protocol relying on recursive swap tests has been proposed and analysed for noisy depolarized states with arbitrary dimension and noise level.
Here, we show that the same algorithm applies much more broadly, enabling the purification of arbitrary mixed states with a nondegenerate principal eigenvalue.
We demonstrate this through two approaches.
In the first approach, we show that, given the largest two eigenvalues, the depolarized noise is the most difficult noise to purify for the recursive swap tests, thus the desirable bounds on performance and cost follow from prior work.  
In the second approach, we provide a new and direct analysis for the performance of purification using recursive swap tests for the more general noise.
We also derive simple lower bounds on the sample complexity, showing that the recursive swap test algorithm attains optimal sample complexity (up to a constant factor) in the low-noise regime.

\end{abstract}

\setlength{\cftbeforesecskip}{4pt}
\renewcommand\cftsecafterpnum{\vskip2pt}
\setcounter{tocdepth}{2}
\tableofcontents

\clearpage

\section{Introduction}

The quantum state purification problem can be stated as follows.  
Let $\ket{\psi} \in \C^d$ be an unknown pure $d$-dimensional state (or \emph{qudit}).
In the \emph{qudit purification problem}, we are given multiple (iid) noisy copies of $\ket{\psi}$ and wish to produce one copy of $\ket\psi$ with less noise.  A particularly
well-studied case is the depolarized state where each noisy copy is a depolarized version of $\ket{\psi}$: 
\begin{equation}
  \rho(\e)
  := (1 - \e) \proj{\psi}
  + \frac{\e}{d} \, I,
  \label{eq:input state 1}
\end{equation}
and $0<\delta<1$. More generally, the noisy copies may take the following form
\begin{equation}
  \tilde{\rho}(\eta)= (1-\eta) \proj{\psi} + \eta \psi^{\perp},
  \label{eq:input state 2}
\end{equation}
where $\psi^{\perp}$ is a state orthogonal to $\proj{\psi}$ and $\ket{\psi}$ is the unique principal eigenvector of $\tilde{\rho}(\eta)$.  
The goal of state purification is to output a quantum state that is closer to $\ket{\psi}$ than the given noisy copies (say, within $\epsilon$ trace distance 
from $\proj{\psi}$).  
Note that removing noise increases the distinguishability of different possible underlying unknown states, which is forbidden for protocols that only use a single copy of the state; thus, purification algorithms necessarily have access to multiple copies of the unknown state.  The quality of the output state is generally improved by having more copies of the noisy state.

In this paper, we are interested in the number of samples that suffice to produce a copy with arbitrarily high fidelity under the more general noise model given in \cref{eq:input state 2}.  As the purification process is intrinsically probabilistic, the number of samples consumed to achieve the task can vary, and we focus on characterizing the expected sample complexity. 
Our problem formulation and noise model follow  \cite{childs2025streaming,GPS24}. These generalize the initial purification task considered by \cite{CEM99,KW01} from purification of a qubit ($2$-dimensional state) to purification of a qudit ($d$-dimensional state).  Note also in the qubit case, the two noise models (\cref{eq:input state 1} and \cref{eq:input state 2}) coincide.

We first summarize some related prior results before stating ours.
In 1998, Cirac, Ekert, and Macchiavello identified the optimal procedure for the purification of a single qubit system \cite{CEM99}. Shortly thereafter, Keyl and Werner considered several variants of the qubit purification problem, with in depth analysis on multiple output copies under various success  measures \cite{KW01}.
Little was known about purification of higher dimensional states until 2016, when
the optimal protocol for 3-dimensional depolarized quantum states was reported and analysed in \cite{futhesis}.  
In 2024, \cite{LFIC24} extended this protocol to depolarized quantum states of arbitrary dimension.
 These optimal protocols rely on weak Schur sampling which requires a collective measurement that can be computationally involved.

In 2023, \cite{childs2025streaming} proposed and analysed a purification protocol for arbitrary dimension for depolarized noise based on recursive swap tests.  
For any constant depolarizing noise strength $\delta < 1$ in \cref{eq:input state 1}, this method has optimal sample complexity $\Theta(1/\epsilon)$; low circuit complexity $O((\log d)/\epsilon)$; and low space complexity $O(\log d \log (1/\epsilon))$. 
In 2024, \cite{GPS24} applied the same protocol to the more general noise model in \cref{eq:input state 2}, and 
claimed, without proof,\footnote{Indeed, this claim was based on work in progress which has come to form part of the present paper.} that a similar sample complexity could be obtained if the principal eigenvalue $1-\eta$ is greater than 1/2.

Our main contribution is a performance analysis of the recursive swap test purification protocol that extends (i) to  the more general noise model in \cref{eq:input state 2} and (ii) to an arbitrarily small spectral gap.
This generalizes the results in both \cite{childs2025streaming,GPS24}.  
This protocol acts only on a few quantum systems at a time, and is computationally less intensive compared to weak Schur sampling, which acts more collectively.  Furthermore, there is a streaming algorithm with very low quantum memory requirement.  

We present two different proofs 
on the rate of convergence of the principal eigenvalue in the recursive swap test purification protocol. These make distinct conceptual contributions. 

Our first proof shows that in a particular sense, the depolarizing noise model in \cref{eq:input state 1} is the hardest case. This allows us to deduce the sample complexity through a reduction to the special case dealt with in \cite{childs2025streaming}. More specifically, we show that under the recursive swap test purification protocol, depolarized states are the most difficult to purify (for a given pair of largest and second largest eigenvalues).
Thus, the sample complexity in \cite{childs2025streaming} for arbitrarily high noise can be applied to the more general noise model (see \cref{eq:input state 2}) for arbitrary spectrum with non-degenerate principal eigenvalue.  

Our second proof presents a new and elementary performance analysis for the recursive swap test. This proof is entirely self contained and tracks the changes in the entire spectrum. The techniques introduced in this proof may provide tighter upper bounds for more specific noise models.

Both approaches establish a similar sample complexity bound. While the sample complexity of purification is quite complicated (see \Cref{thm:proof1_sample_complexity} and \Cref{thm:proof2_sample_complexity}), we present a simplified version below to give some intuitive sense of the general magnitudes:

\begin{theorem}[simplified]
The expected sample complexity of purification under the general noise model is $4^{O(1/(\lambda_1 - \lambda_2))}(1-\lambda_1)/\epsilon$, 
where $\lambda_1$ and $\lambda_2$ are the top two eigenvalues of the given unknown noisy state.
\end{theorem}

Moreover, we provide simple worst-case lower bounds for the sample complexity for purifying general states, relying on optimal results on depolarized states \cite{CEM99,LFIC24} and following a proof technique similar to that in \cite{childs2025streaming}.
Our result shows that the recursive swap test algorithm achieves optimal sample complexity (up to a constant factor) in the low noise regime.

We mention that very recently, version 2 of \cite{LFIC24} further reports on independent results showing an asymptotically optimal tradeoff between output fidelity and sample complexity for the more general noise model in \cref{eq:input state 2}. This relies on a Schur sampling protocol.
In comparison,
for constant noise, 
our method requires small quantum memory and quite small circuits, and is sample optimal (up to a constant factor), but it becomes suboptimal as the given noisy state becomes close to maximally mixed (or as the spectral gap vanishes).  For the very noisy regime, a third purification method using state tomography remains  efficient in the sample complexity in terms of the dependence on the noise paramater, but becomes inefficient with growing  dimension.  These three purification methods are apparently incomparable. There is no known purification protocol that encompasses all the advantages. A sample and circuit efficient purification may involve an initial collective stage if the noise parameter is high, which acts on multiple but not very large number of copies, followed by a more circuit efficient method such as the recursive swap purification protocol, once the initial noise becomes lower.

There are other variations of the quantum state purification problem.  We mention a few here for contrast, but note that they are out of the current scope.  
One variation of the problem revolves around the tradeoff between the sample complexity, fidelity, \emph{and the probability of successful purification}.  Fruitful studies have been made using semidefinite programming and block encoding techniques, such as those reported in \cite{F04,YCHCW24}.  
Another variation is concerned with error suppression wherein the performance is analysed as order expansion in the noise parameter, thus these approaches mostly apply to very low noise; see for example, \cite{YKLO24}.  
Another variation is broadly under the subject of virtual purification or error mitigation. In this line of research, instead of providing an approximation of $\ket{\psi}$ as a quantum state that is used by other quantum information processing task, one argues that these tasks end with measurements, and it suffices to have a protocol that approximates the final measurement statistics.  Examples includes improved methods for quantum metrology \cite{YESHMT22}, and extensions for purification of channels \cite{LZFC24}.

The paper is structured as follows.  In \cref{sec:swap}, we describe the swap test and the purification protocol based on recursive swap test and summarize known performance results in \cite{childs2025streaming,GPS24}.  In \cref{sec:analysis}, we present two analysis on the performance of recursive swap test for the more general noise model \cref{eq:input state 2}, showing that the convergence of the principal eigenvalue is no slower than for the first noise model \cref{eq:input state 1}. Our result extends various applications of quantum state purification discussed in \cite{childs2025streaming,GPS24} to higher noise regime or more general noise model.  
In \cref{sec:lowbdd} we present our lower bounds for the sample complexity for purifying general states.

\section{Purification using the swap test}\label{sec:swap}

We summarize the swap test, the swap test gadget, and the quantum state purification protocol 
based on the recursive swap tests in this section.  The proofs of the main claims in this section are elementary to verify and can also be found in \cite{childs2025streaming}.

\subsection{The swap test} \label{sec:swaptest}

The swap test, with two $d$-dimensional inputs and one $d$-dimensional output, 
is given by the circuit in \cref{fig:swap}, which we adapt from \cite{childs2025streaming}.  
\begin{figure}[ht]
\centering
\tikzset{trashcan/.style={path picture={
		\coordinate (bin) at ($(path picture bounding box.center)+(0,-0.25cm)$);
		\draw[internal,inner sep=0pt,-stealth,thickness,rounded corners] (path picture bounding box.west) -| (bin);
		\draw[thickness] (bin) circle[x radius = 0.2cm, y radius = 0.1cm] ++(0.2cm,0) -- ++(-0.06cm,-0.35cm) arc [start angle = 0, end angle = -180, x radius = 0.14cm, y radius = 0.07cm] -- ++(-0.06cm,0.35cm);	
		},
		minimum height=3.6em,
		minimum width=2em}}
\begin{quantikz}[row sep={0.7cm,between origins},column sep=0.4cm]
    \lstick{$\ket{0}$} & \gate{H} & \ctrl{2} & \gate{H} & \meter{} & \setwiretype{c} \rstick{$a$} \\
     & & \targX{} & & & \\
     & & \targX{} & & |[trashcan]| {} & \setwiretype{n} \
\end{quantikz}
\caption{\label{fig:swap} Quantum circuit for the swap test, with two inputs on the left, and one output on the right.  The gate in the middle with two crosses denotes the controlled-swap gate, while $H$ labels the Hadamard gate.  The measurement box on the first qubit is along the computational basis, while the trash bin on the last system represents a partial trace.}
\end{figure}
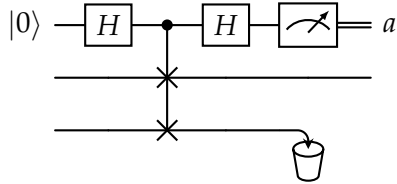

Let $X$ be the swap operator.  If the joint input is $\mu$ and the measurement outcome is $a \in \set{0,1}$ in \cref{fig:swap}, the sub-normalized post-measurement state (indicated by the prime) on the remaining registers is
\begin{equation}
  \SWAP'(\mu,a) = \frac{I + (-1)^a X}{2} \; \mu \; \frac{I + (-1)^a X}{2} \,.
  \label{eq:spoutboth}
\end{equation}
The probability of obtaining the outcome $a$ is given by the 
trace of the above.  
Furthermore, if $\mu = \rho \x \sigma$ is in product form, 
each outcome happens with probability 
\begin{equation}
  \frac{1}{2} \of[\big]{1 + (-1)^a \Tr(\rho\sigma)} \,.
  \label{eq:trace}
\end{equation}
Discarding the bottom register in the circuit gives the subnormalized output state 
\begin{equation}
	\SWAP'(\rho \otimes \sigma,a) = \frac{1}{4} \of[\big]{\rho + \sigma + (-1)^a \rho \sigma + (-1)^a \sigma \rho}.
  \label{eq:final_state}
\end{equation}
We are mostly interested in the output state corresponding to the measurement outcome $a = 0$: 
\begin{equation}
  \SWAP(\rho \otimes \sigma)
  := \frac{1}{2} \cdot
     \frac{\rho + \sigma + \rho \sigma + \sigma \rho}
          {1 + \Tr(\rho\sigma)},
  \label{eq:SWAP}
\end{equation}
where we omit the label $a=0$ and also remove the prime to indicate the normalized output state.
\Cref{alg:swap} (adapted from \cite{childs2025streaming} with minor modification) specifies a \emph{$\SWAP$ gadget} that performs the swap test on fresh copies of the available input state until it succeeds in obtaining the $a = 0$ outcome.

\begin{algorithm}[H]
\caption{$\normalfont\textsc{Swapgadget}$ (adapted from \cite{childs2025streaming}).}
\label{alg:swap}
\emph{This procedure has access to a stream of inputs $\mu$, and uses as many copies of them as necessary.}
\begin{algorithmic}[1]
\Procedure{$\normalfont\textsc{Swapgadget}$}{$\mu$}
\Repeat
  \State Ask for a fresh copy of $\mu$.
  \State Apply the swap test shown in \cref{fig:swap} and denote the measurement outcome by $a$.
\Until $a=0$
\State \Return the state of the output register
\EndProcedure
\end{algorithmic}
\end{algorithm}

Here, we mostly consider applying the swap test when $\rho = \sigma$.    
From \cref{eq:trace}, the probability of a successful swap test (i.e. observing the outcome $a = 0$) is given by
\begin{equation}
  \Pr(a=0~|~\rho)=\frac{1}{2} \of[\big]{1 + \Tr (\rho^2)}.
  \label{eq:prob}
\end{equation}
In this case, the normalized output state is
\begin{equation}
  \SWAP(\rho^{\otimes 2}) = \frac{\rho + \rho^2}{1+\Tr(\rho^2)}.
  \label{eq:rout}
\end{equation}
Note that $\rho$, $\rho^2$ and hence also $\SWAP(\rho^{\otimes 2})$ have identical eigenvectors. Additionally, if the eigenvalues of $\rho$ are $\vlambda=(\lambda_1,\ldots,\lambda_d)$ with $\lambda_1 > \lambda_2 \geq \lambda_3 \geq \cdots \geq \lambda_d$, 
then the eigenvalues of $\SWAP(\rho^{\otimes 2})$ are 
\begin{equation}
\lambda_i' = \frac{\lambda_i + \lambda_i^2}{1+\norm{\vlambda}^2} \,,  
\label{eq:evals}
\end{equation}
where $\norm{\vlambda}=\sqrt{\sum_{j=1}^d \lambda_j^2}$ is the $2$-norm of the spectrum of $\rho$.
Our analysis relies heavily on two immediate consequences of \cref{eq:evals}.  
First, eigenvalue ordering is preserved under the swap test i.e. $\lambda_1' > \lambda_2' \geq \lambda_3' \geq \cdots \geq \lambda_d'$.  
Second,
\begin{align}
\lambda_1' - \lambda_1 
& = \frac{\lambda_1 + \lambda_1^2}{1+\norm{\vlambda}^2} - \lambda_1
= \frac{\lambda_1^2  \sum_{j=1}^d \lambda_j  - \lambda_1 \sum_{j=1}^d \lambda_j^2}{1+\norm{\vlambda}^2} 
= \frac{ \lambda_1  \sum_{j=2}^d (\lambda_1 - \lambda_j) \lambda_j }{1+\norm{\vlambda}^2} 
\label{eq:purer} 
\end{align}
which is positive unless $\rho$ is pure or maximally mixed -- in other words, if $a=0$, the 
weight of the principal eigenvector is strictly increased by the swap test for \emph{all}  
non-trivial noisy inputs.  
A recursive algorithm that applies the $\SWAP$ gadget on a pair of identical inputs produced in the previous level of the recursion thus makes the output purer at each subsequent level of recursion.  

As a side remark, note that \cref{eq:evals} applies to any $\rho$.  In particular, degenerate eigenspaces are preserved.  While we focus on noisy input $\rho$ with non-degenerate principal eigenvalue, \cref{eq:purer} states that, as long as $\rho$ is not maximally mixed, $\SWAP(\rho^{\otimes 2})$ will have strictly higher weight in the principal eigenspace compared to $\rho$, and recursive applications will suppress all other eigenvectors, providing a potentially useful generalization of the purification task for future study.

\subsection{Recursive purification based on the swap test} \label{sec:recursive}

Given access to many copies of $\rho_0$ which is either $\rho(\e)$ from \cref{eq:input state 1} or $
\tilde{\rho}(\eta)= (1-\eta) \proj{\psi} + \eta \psi^{\perp}$ from \cref{eq:input state 2}, consider the following recursive algorithm (adapted from \cite{childs2025streaming} with minor modification): 
\begin{algorithm}[H]
\caption{Recursive purification based on the swap test (adapted from \cite{childs2025streaming}).}
\label{alg:purify}
\emph{This procedure purifies a stream of states $\rho_0$ by recursively calling the $\SWAP$ gadget (\cref{alg:swap}).}
\begin{algorithmic}[1]
\Procedure{$\PURIFY$}{$n$}
\If{$n=0$}
\Comment{If at the bottom level of recursion,}
  \State \Return $\rho_0$
  \Comment{request one copy of $\rho_0$ from the stream and return it.}
\Else
\Comment{Otherwise call the next level of recursion}
  \State \Return 
  $\SWAPGADGET(\PURIFY(n-1)\otimes\PURIFY(n-1))$
  \Comment{until the $\SWAP$ gadget succeeds.}
\EndIf
\EndProcedure
\end{algorithmic}
\end{algorithm}

In other words, copies of $\rho_0$ are used by the swap gadget to produce copies of $\rho_1 = \SWAP(\rho_0^{\otimes 2})$, which in turns are used by the swap gadget to produce copies of $\rho_2 = \SWAP(\rho_1^{\otimes 2})$, and so on.  We recursively apply this for $n$ levels, until 
$\frac{1}{2}\| \rho_n - \proj{\psi} \|_1 \leq \epsilon$ for some desirable output noise level $\epsilon \ll \e$.  
Note that a call to $\SWAPGADGET(\mu)$ uses more than one copy of $\mu$ when some of the swap test rounds result in failure (i.e. $a=1$). Specifically, the expected number of copies of $\mu$ used in the call is $1\leq 1/\Pr(a=0 \mid \PURIFY(j-1)) \leq 2$ (see \cref{eq:prob}). The lower bound is saturated when $\mu$ is a tensor product of two identical pure states while the upper bound is saturated in the limit of $d\rightarrow \infty$ by a pair of maximally mixed states with dimension d. 
Whenever a swap test fails, the recursion is restarted at that level, which in turns calls for the previous level, and so on, cascading down in level until $\PURIFY(0)$'s are called which then directly use copies of the given noisy state $\rho_0$.

In the above algorithm, all swap tests act on the tensor product of two identical states, so, the eigenspaces of $\rho_0$ are preserved at all levels.  Thus the trace distance of any output state from the principal eigenvector is simply $1$ minus the principal eigenvalue of the output, and we focus on the latter.

Let $\rho_k := \PURIFY(k)$.  For depolarized noise, $\rho_k = \rho(\e_k)$ 
where 
\begin{align}
  \e_k &:= \frac{\e_{k-1} + \e_{k-1}^2/d}{2 p_{k-1}},  
  &
  p_{k-1} &:= 1 - \of[\Big]{1 - \frac{1}{d}} \e_{k-1} + \frac{1}{2} \of[\Big]{1 - \frac{1}{d}} \e_{k-1}^2,  
  \label{eq:PD}
\end{align}
and $p_k$ is the probability of getting outcome $a = 0$ when the swap test is applied on two copies of $\rho_{k-1}$. Setting the initial value $\e_0 = \e$, the above gives a recurrence for $\e_k$ and $p_k$.  For more general noisy input states \cref{eq:input state 2}, the recurrence for $p_k$ and for the principal eigenvalue of $\rho_k$ can be defined similarly using \cref{eq:prob,eq:rout}.  

Reference  \cite{childs2025streaming} also presents a stack-based implementation for $\PURIFY(n)$ in \cref{alg:purify} so that it uses only $n+1$ qudits of quantum memory plus one ancilla qubit.

\section{Analysis for general noise}\label{sec:analysis}

Our main result is, loosely, that the recursive swap test algorithm performs about as well with general noise as it does with depolarizing noise. We found two proofs of this fact with different advantages,
so we present them both. 

The first proof, in \Cref{sec:proof1}, has the style of a reduction, arguing that the problem can be reduced to analyzing the depolarizing case, albeit with a slightly generalized understanding of spectra. The second approach (\Cref{sec:proof2}) is a more typical stand-alone proof of the result.  

Before we present these proofs, we discuss some important notions concerning the spectrum of density matrices that are crucial for our proofs.  

\subsection{Preliminaries on spectrum}\label{sec:preliminaries}

Purification algorithms (and especially the swap test algorithm) are concerned with the \emph{spectrum} of the input state $\rho$, i.e., the vector of eigenvalues $\vlambda = (\lambda_1, \ldots, \lambda_d)$ in descending order. 

\begin{definition}
    In dimension $d \geq 2$, define the space of valid spectra, $V \subseteq \mathbb R_{\geq 0}^{d}$, as the set of non-negative real $d$-dimensional vectors $\vlambda = (\lambda_1, \ldots, \lambda_d)$ such that 
\begin{enumerate}
    \item $\vlambda$ is non-increasing,
    \item $\lambda_1 > \lambda_2$,
    \item $\sum_{i=1}^d \lambda_i=1$.
\end{enumerate}
\end{definition}
\noindent Notice that the definition requires a gap between $\lambda_1$ and $\lambda_2$, otherwise it is unclear which eigenspace purification should amplify. In fact, the gap is important for analysis, so we define $\Delta(\vlambda) := \lambda_1 - \lambda_2$, or simply $\Delta$ when there is no ambiguity what $\vlambda$ refers to. 

Since the algorithm repeatedly applies the swap test to identical pairs of input state, it is absolutely critical to quantify two things: the probability that the test succeeds and the spectrum of the output state conditioned on success.  
\begin{definition}
Let $v \colon V \to \mathbb R$ be the function $v(\vlambda) = \frac{1 + \norm{\vlambda}^2}{2}$.
\end{definition}
\begin{definition}
Define the map $F \colon V \to V$ such that 
for all $\vlambda \in V$, the output $F(\vlambda)$ has  
$i$-th entry given by
\begin{equation}
\label{eq:defnf}
F(\vlambda)_i 
:= \frac{\lambda_i (\lambda_i + 1)}{1 + \norm{\vlambda}^2}. 
\end{equation}
\end{definition}
\begin{proposition}
Let $\rho$ be a state with spectrum $\vlambda$. A swap test on two copies of $\rho$ succeeds with probability $v(\vlambda)$, and then discarding a copy leaves a state with spectrum $F(\vlambda)$ over the same eigenbasis as $\rho$.  Furthermore, $F(\vlambda)_1 > \lambda_1$.  
\end{proposition}
\begin{proof}
This follows from \cref{eq:prob}, \cref{eq:evals}, and \cref{eq:purer} 
\end{proof}

\subsection{Proof 1 -- Reduction from general spectrum to depolarizing noise}\label{sec:proof1}

In our first proof, we argue that for given $\lambda_1$ and $\lambda_2$, the worst case setting in recursive swap test purification corresponds to the depolarizing noise spectrum (where $\lambda_d=\lambda_2$). 
We wish to ensure generality of our result for all $\lambda_1>\lambda_2$ with $\lambda_1+\lambda_2\leq 1$. 
However, note that the dimension of a depolarized  state is determined by $\lambda_1$ and $\lambda_2$ as $d=1+(1-\lambda_1)/\lambda_2$.
This means that we need to compare spectra of systems with potentially different dimension.
Moreover, we need to permit choices with positive non-integer values of $(1-\lambda_1)/\lambda_2$ which for a depolarizing noise spectrum would ordinarily correspond to $d-1$. We achieve this by introducing a generalized notion of spectra and show that within this framework, for a given $\lambda_1$ and $\lambda_2$, the generalized spectrum consistent with depolarizing noise is the worst case setting in sample complexity. That is, it requires the most expected samples to purify using the swap test.

After defining generalized spectra (c.f. \cref{def:gen_spec}), we define a projector $S$ that sends any spectrum $\blambda$ to the unique generalized depolarizing noise spectrum with matching leading spectral elements $\lambda_1$ and $\lambda_2$. We define the function $F$ that maps any generalized spectrum to a new generalized spectrum corresponding to the spectrum of the output state of a successful swap test. In \cref{lem-compare}, we show that for all $k=0,1,2 ,\ldots$, and for all $\blambda$, the spectrum $F^{(k)}(\blambda)$ is at least as pure as $F^{(k)}(S(\blambda))$ where $F^{(k)}$ represents the $k$-fold application of the $F$ map. Finally, we observe that the results from \cite{childs2025streaming} showing the sample complexity associated with purification of depolarized states also apply to generalized depolarized states. We utilize this to derive upper bounds on the sample complexity for any state that has a spectrum with a given leading pair of eigenvalues.

\subsubsection{Generalized Spectra}

Technically, not all $\lambda_1$ and $\lambda_2$ correspond to a depolarizing-noise state.
In addition to the obvious constraints $\lambda_1+\lambda_2\leq 1$ and $\lambda_1\geq \lambda_2\geq 0$, the dimension $d$, given by $d=1 + \frac{1-\lambda_1}{\lambda_2}$ for a depolarizing-noise state needs to be a positive integer. 
Hence, we generalize the notion of spectrum to allow eigenvalues to occur a non-integer number of times. 
\begin{definition}[generalized spectra]\label{def:gen_spec}
A pair of real vectors, denoted as $\Lambda=(\bp,\blambda)$, where $\bp=(1,q_2,\cdots)$ and $\blambda=(\lambda_1,\lambda_2,\cdots)$, is called a \emph{generalized spectrum}, if
\begin{itemize}
    \item $q_1=1$, $q_i\geq 0~(\forall i)$; 
    \item $\lambda_1 > \lambda_2\geq\cdots\geq0$; 
    \item $\sum q_i\lambda_i=1$.
\end{itemize}
\end{definition}
Here, $\bp$ and $\blambda$ have the same length, which must be finite but is otherwise not so important because the vectors can be padded with zeros without changing the meaning.

\begin{definition}
    For a generalized spectrum $\Lambda=(\bp,\blambda)$, we define its gap as:
\begin{align}
    \Delta(\Lambda)=\lambda_1-\lambda_2.
\end{align}
We define two other generalized spectra $S(\Lambda)$ and $F(\Lambda)$ as:
\begin{align}
    S(\Lambda)&=((1,\frac{1-\lambda_1}{\lambda_2}), (\lambda_1,\lambda_2)),\label{eq-def-S}
    \\
    F(\Lambda)&=(\bp,\frac{1}{1+\sum q_i\lambda_i^2}(\lambda_1+\lambda_1^2,\lambda_2+\lambda_2^2,\cdots)). \label{eq-def-F}
\end{align}
\end{definition}
The fact that $S(\Lambda)$ and $F(\Lambda)$ satisfy the definition of generalized spectra are easy to check.
Denote $\mathcal{S}$ as the image of $S$.
Equivalently, $\mathcal{S}$ is the set of generalized spectra for which the dimensions of $\bp$ and $\blambda$ are 2:
\begin{equation}\label{eq-defSS}
    \mathcal{S}=\{((1,q_2),(\lambda_1,\lambda_2)) \mid
    q_2\geq0, \lambda_1 > \lambda_2\geq 0, \lambda_1+q_2\lambda_2=1\}.
\end{equation}
It is obvious that $S^2=S$, hence $S$ is a projection onto $\mathcal{S}$.

In the following $\lambda_1$ and $\Delta$ will be the most important parameters to characterize the hardness of purification.
The allowable space for $\lambda_1$ and $\Delta$ is:
\begin{equation}
    0< \Delta \leq \lambda_1 \leq 1.
\end{equation}
For generalized spectra correspond to eigenspectra of physical density matrices, $\lambda_1+\lambda_2\leq 1$, hence we also have:
\begin{equation}
    \Delta\geq 2\lambda_1-1.
\end{equation}

\begin{definition}[ordering]
    For two generalized spectra $\Lambda$ and $\Lambda'$, define $\Lambda \succeq \Lambda'$ if $\lambda_1\geq \lambda_1'$ and $\Delta(\Lambda)\geq \Delta(\Lambda')$. 
\end{definition}
Here, the dimensions (lengths) of $\Lambda$ and $\Lambda'$ do \emph{not} need to be equal, since we only care about the first two entries of $\vlambda$.
The relation $\succeq$ satisfies reflexivity and transitivity.
It is \emph{not} a partial order in general since it does not satisfy antisymmetry. 
It is a partial order on $\mathcal{S}$.

\subsubsection{Technical Lemmas}

\begin{lemma}\label{lem-FkeepS}
    If $\Lambda\in \mathcal{S}$  then $F(\Lambda)\in \mathcal{S}$.
\end{lemma}
\begin{proof}
    The conditions of \eqref{eq-defSS} follow straightforwardly:
    \begin{enumerate}
        \item $F(\Lambda)$ does not change the $\bp$ vector, so $q_1 = 1$ and $q_2$ is nonzero. 
        \item $F(\Lambda)$ is a generalized spectrum, so the eigenvalues are descending.
        \item $\lambda_1 + q_2 \lambda_2 = 1$ follows from simple algebraic manipulation of \eqref{eq-def-F}. \qedhere
    \end{enumerate}
\end{proof}

\begin{lemma}\label{lem-Skeeporder}
    If $\Lambda\succeq \Lambda'$, then $S(\Lambda)\succeq S(\Lambda')$.
\end{lemma}
\begin{proof}
    Notice that $S$ preserves $\lambda_1$ and $\lambda_2$, and thus also $\Delta$. The inequalities in $\lambda_1$ and $\Delta$ implied by $\Lambda \succeq \Lambda'$ thus carry over to $S(\Lambda)$ and $S(\Lambda')$, giving us $S(\Lambda) \succeq S(\Lambda')$.
\end{proof}

\begin{lemma}\label{lem-Fkeeporder}
    If $\Lambda, \Lambda'\in \mathcal{S}$ and if $\Lambda\succeq \Lambda'$, 
    then $F(\Lambda)\succeq F(\Lambda')$.
\end{lemma}
\begin{proof}
    For $\Lambda\in\mathcal{S}$, we have $q_2=\frac{1-\lambda_1}{\lambda_2}$, $\lambda_2=\lambda_1-\Delta$.
    Therefore, $$1+\sum q_i\lambda_i^2 = 1 + \lambda_1^2 + \frac{1-\lambda_1}{\lambda_2} \lambda_2^2 = 1 + \lambda_1^2 + \lambda_2 - \lambda_1 \lambda_2 = 1 + \lambda_2 + \lambda_1 \Delta = 1 + \lambda_1 - (1-\lambda_1) \Delta.$$  
    Substituting this into \cref{eq-def-F},
    \begin{align}
        \lambda_1(F(\Lambda))&=\frac{\lambda_1+\lambda_1^2}{1+\lambda_1-(1-\lambda_1)\Delta},\\
        \Delta(F(\Lambda))&=\frac{\Delta(1+2\lambda_1-\Delta)}{1+\lambda_1-(1-\lambda_1)\Delta}.
    \end{align}
    Explicit calculation shows that:
    \begin{align}
        \frac{\partial \lambda_1(F(\Lambda))}{\partial\lambda_1}
        &=\frac{(1-\Delta)(1+2\lambda_1)+(1+\Delta)\lambda_1^2}{(1+\lambda_1-(1-\lambda_1)\Delta)^2}\geq 0 \,,\\
        \frac{\partial \lambda_1(F(\Lambda))}{\partial\Delta}
        &=\frac{\lambda_1(1-\lambda_1^2)}{(1+\lambda_1-(1-\lambda_1)\Delta)^2}\geq 0 \,,\\
        \frac{\partial \Delta(F(\Lambda))}{\partial\lambda_1}
        &=\frac{(1-\Delta)^2\Delta}{(1+\lambda_1-(1-\lambda_1)\Delta)^2}\geq0 \,,\\
        \frac{\partial \Delta(F(\Lambda))}{\partial\Delta}
        &=\frac{
        (1-\Delta)^2+(3-2\Delta-\Delta^2)\lambda_1+2\lambda_1^2
                }{(1+\lambda_1-(1-\lambda_1)\Delta)^2}\geq 0 \,.\label{eq-monoDeltaDelta}
    \end{align}
    In the last inequality, we use the fact that $0<\Delta\leq 1$.
    Therefore $F(\Lambda)\succeq F(\Lambda')$.
\end{proof}

\begin{lemma}\label{lemma1}
    $F(\Lambda)\succeq F(S(\Lambda))$.
\end{lemma}
\begin{proof}
    Due to \cref{eq-def-S}, $\Lambda$ and $S(\Lambda)$ have the same $\lambda_1$ and $\lambda_2$ component.
    Therefore, it suffices to show
    \begin{equation}
        \sum q_i\lambda_i^2\leq \lambda_1^2+\frac{1-\lambda_1}{\lambda_2}\lambda_2^2.
    \end{equation}
    This holds since 
    \begin{equation}
            \sum q_i\lambda_i^2 = \lambda_1^2+\sum_{i\geq 2} q_i\lambda_i^2 \leq \lambda_1^2+\lambda_2\sum_{i\geq 2} q_i\lambda_i = \lambda_1^2+\lambda_2(1-\lambda_1), 
    \end{equation}
    where we have used the fact that $\lambda_2\geq \lambda_i~(\forall i\geq2)$ and $\sum q_i\lambda_i=1$.
\end{proof}

\begin{lemma}\label{lem-compare}
    $\forall k\in\mathbb{N}$, $F^{(k)}(\Lambda)\succeq F^{(k)}(S(\Lambda))$. Here $F^{(k)}$ denotes iterating $F$ for $k$ times.
\end{lemma}
\begin{proof}
    We prove this lemma by induction on $k$.
    The case of $k=0$ is trivial. 
    The case of $k=1$ is \cref{lemma1}.
    Now assume the induction hypothesis holds for some $k \in \mathbb{N}$, namely, $F^{(k)}(\Lambda)\succeq F^{(k)}(S(\Lambda))$. Further applying \cref{lem-Skeeporder}, we have 
    \begin{equation}
        S(F^{(k)}(\Lambda))\succeq S(F^{(k)}(S(\Lambda)))=F^{(k)}(S(\Lambda)),
    \end{equation}
    where the last equality follows from \cref{lem-FkeepS}.
    Applying \cref{lem-Fkeeporder}, we have:
    \begin{equation}
        F(S(F^{(k)}(\Lambda)))\succeq  F(F^{(k)}(S(\Lambda)))=F^{(k+1)}(S(\Lambda)).
    \end{equation}
    On the other hand, due to \Cref{lemma1}, 
    \begin{equation}
        F^{(k+1)}(\Lambda)=F(F^{(k)}(\Lambda))\succeq F(S(F^{(k)}(\Lambda))),
    \end{equation}
    hence
    \begin{equation}
        F^{(k+1)}(\Lambda)\succeq F^{(k+1)}(S(\Lambda)).
    \end{equation}
    This completes the induction.
\end{proof}

\subsubsection{Upper Bound on Iteration Levels}

For a $d$-dimensional mixed state $\rho$, we associate a generalized spectrum as
\begin{equation}\label{eq-lambdaofrho}
    \Lambda(\rho)=((1,\cdots,1),(\lambda_1,\cdots,\lambda_d)),
\end{equation}
where $\lambda_i$ are the eigenvalues of $\rho$ in descending order.
Recall from \cref{eq:rout}
that on copies of $\rho$, the swap gadget returns the normalized state:
\begin{equation}
  \SWAP(\rho^{\otimes 2}) = \frac{\rho + \rho^2}{\Tr(\rho + \rho^2)}.
\end{equation}
Therefore,
\begin{equation}
    \Lambda(\SWAP(\rho^{\otimes 2}))=F(\Lambda(\rho)).
\end{equation}
For purification, we care about how $\lambda_1(F^{(k)}(\Lambda(\rho)))$ converges to 1 as $k$ increases. \Cref{lem-compare} says that $F^{(k)}(\Lambda)\succeq F^{(k)}(S(\Lambda))$, so $\lambda_1(F^{(k)} (\Lambda(\rho))  )$ is bounded below by $\lambda_1(F^{(k)}(S(\Lambda(\rho))))$, so it suffices to consider the trajectory starting from $S(\Lambda(\rho))$.
Operationally, this signifies that 
$\rho$ is easier to purify than a "depolarized state with generalized spectrum" having the same largest and second largest eigenvalues.

A generalized spectrum $\Lambda=(\bp,\blambda)\in\mathcal{S}$ is determined by two independent parameters. 
We make a change of variable to $d_{\eff}$ and $\delta$ by  
\begin{align}
    d_{\eff}&=1+\frac{1-\lambda_1}{\lambda_2}=\frac{1-\Delta}{\lambda_1-\Delta},\\
    \delta&=d_{\eff}\lambda_2=1-\Delta.
\end{align}
Note that if $\Lambda$ comes from a physical mixed state \cref{eq-lambdaofrho}, then
\begin{equation}
     2\leq d_\eff\leq d,
\end{equation}
where the first equality holds if and only if $\lambda_3=\cdots \lambda_d=0$, and the second equality holds if and only if $\rho$ takes the form of depolarized states.
Using these variables, we have 
\begin{equation}
    q_2=d_{\eff}-1,~~
    \lambda_1=1-\delta+\frac{\delta}{d_{\eff}},~~
    \lambda_2=\frac{\delta}{d_{\eff}}.
\end{equation}
We may think of $d_{\eff}$ as a (possibly non-integer) dimension, and $\delta$ as the noise parameter. 
In fact, when $d\in\mathbb{N}^+$, the spectra of $(1-\delta)\ket{\psi}\bra{\psi}+\frac{\delta}{d}I_d$ is exactly given by the above equations.
The function $F$, when restricted to $\mathcal{S}$, can be re-expressed as:
\begin{align}
    d_{\eff}(F(\Lambda))&=d_{\eff}(\Lambda),\\
    \delta(F(\Lambda))&=
    \frac{\delta + \delta^2/d_{\eff}}{2 - 2(1 - \frac{1}{d_{\eff}}) \delta + (1 - \frac{1}{d_{\eff}})\delta^2}.
\end{align}
This is exactly the same as \cref{eq:PD},
with $d$ there replaced by $d_{\eff}$.
Therefore, $F^{(k)}(S(\Lambda))$ undergoes the same dynamics as in the depolarizing noise case. 
Thanks to \Cref{lem-compare}, any upper bound for the level of iterations needed in the depolarizing noise case, if the proof of which does not require $d$ to be an integer, gives rise to an upper bound for the level of iterations needed for the general case.

\begin{theorem}
\label{thm:proof1_round_complexity}
    To purify a mixed state $\rho$ via recursive swap test until the final principal eigenvalue is larger than $1-\epsilon$, 
    it suffices if the algorithm recurses to the following number of levels:
    \begin{equation}\label{eq:Nupperbound}
        \begin{cases}
            \log_2\frac{1-\lambda_1}{(2\Delta-1)\epsilon}, & \text{if~} \Delta\geq \frac{2}{3} \,,\\
            \log_2\frac{1-\lambda_1}{(1-\Delta)\epsilon}+5, & \text{if~} \frac{1}{3}\leq \Delta<\frac{2}{3} \,,\\
            \log_2\frac{1-\lambda_1}{(1-\Delta)\epsilon}+5+
            f(\lambda_1,\Delta), & \text{if~} 0<\Delta<\frac{1}{3} \,.
        \end{cases}
    \end{equation}
    Here, $\lambda_1=\lambda_1(\rho)$, $\lambda_2=\lambda_2(\rho)$,  $\Delta=\lambda_1-\lambda_2$, and
    \begin{equation}\label{eq:defoff}
        f(\lambda_1,\Delta)=
        \min\left\{
        \frac{1}{\Delta}+2\nlog\frac{1}{\Delta},
        \of*{\frac{1-\lambda_1}{\lambda_2}+3}\nlog\left(1+\frac{\lambda_2}{\Delta(1-\Delta)}\right)+5\nlog\frac{1-\Delta}{\lambda_2}
        \right\}.
    \end{equation}
\end{theorem}
    \noindent Note that $\frac{1-\lambda_1}{2\Delta-1}\leq 1$ when $\Delta\geq \frac{2}{3}$ and $\frac{1-\lambda_1}{1-\Delta}\leq 1$,
    hence the first logarithms in each line are always upper bounded by $\log_2\frac{1}{\epsilon}$.
\begin{proof}
    For a state $\rho$, we define $\Lambda(\rho)$ via \eqref{eq-lambdaofrho} and hence
    \begin{equation}\label{eq-initialSLambda}
    \begin{aligned}
        S(\Lambda(\rho))&=((1,\frac{1-\lambda_1}{\lambda_1-\Delta}), (\lambda_1,\lambda_1-\Delta)),\\
        d_{\eff}(S(\Lambda(\rho)))&=\frac{1-\Delta}{\lambda_1-\Delta},\\
        \delta(S(\Lambda(\rho)))&=d_{\eff}\lambda_2=1-\Delta.
    \end{aligned}
    \end{equation}
    To purify the state so that $\lambda_1(F^{(k)}(\Lambda(\rho)))\geq1-\epsilon$, due to \Cref{lem-compare}, it is enough to demand
    \begin{equation}
        \lambda_1(F^{(k)}(S(\Lambda(\rho))))\geq1-\epsilon,
    \end{equation}
    or equivalently
    \begin{equation}\label{eq:finaldelta}
        \delta^{(k)}\leq\frac{d_\eff^{(k)}\epsilon}{d_\eff^{(k)}-1}.
    \end{equation}
    Here, $\delta^{(k)}=\delta(F^{(k)}(S(\Lambda(\rho))))$ and $d_\eff^{(k)}$ is defined similarly.
    From the discussions above, we have $d_\eff^{(k)}=d_\eff^{(0)}$, which is given by \cref{eq-initialSLambda}. 
    Hence the above equation becomes
    \begin{equation}
        \delta^{(k)}\leq\frac{(1-\Delta)\epsilon}{1-\lambda_1}.
    \end{equation}
    
    The theorem is then proved by following the discussions in \cite{childs2025streaming} (Section 3.2), since the proofs therein do not require $d$ to be an integer\footnote{Except claim $2^{(d)}$, which is required by lemma 8. But it can be easily modified to allow arbitrary $d\in [2,\infty)$.}.
    More precisely: 
    
    If $\delta^{(0)}\leq \frac{1}{3}$ ($\Delta\leq \frac{2}{3}$), then eq(31) in \cite{childs2025streaming} shows that 
    $k=\log_2\frac{\delta^{(0)}}{(1-2\delta^{(0)})\delta^{(k)}}=\log_2\frac{1-\lambda_1}{(2\Delta-1)\epsilon}$ suffices.

    If $\frac{1}{3}<\delta^{(0)}\leq \frac{2}{3}$ ($\frac{1}{3}\leq \Delta<\frac{2}{3}$), then \cite[Section 3.2.2]{childs2025streaming} shows that $\log_2\frac{1}{\delta^{(k)}}+5$ iterations suffice.

    If $\delta^{(0)}>\frac{2}{3}$ ($\Delta<\frac{1}{3}$), then Lemma 7 in \cite{childs2025streaming} shows that $\log_2\frac{1}{\delta^{(k)}}+\frac{1}{1-\delta^{(0)}}+2\nlog \frac{1}{1-\delta^{(0)}}$ iterations suffice.
    Lemma 8 in \cite{childs2025streaming} shows that 
    $\log_2\frac{1}{\delta^{(k)}}+(\nlog(1+\frac{1}{d_\eff+1}))^{-1}\nlog\left(1+\frac{1}{d_\eff(1-\delta^{(0)})}\right)+ 5\nlog d_\eff$ also suffices.

\end{proof}

\subsubsection{Sample Complexity Upper Bound}

In this part, we analysis the sample complexity of our algorithm. Our result is as follows.
\begin{theorem}
\label{thm:proof1_sample_complexity}
    To purify a mixed state $\rho$ via recursive swap test until the final principal eigenvalue is larger than $1-\epsilon$, the expected sample complexity is at most
    \begin{equation}\label{eq:Nupperboundthm1}
        \begin{cases}
            \frac{9(1-\lambda_1)}{\epsilon}, & \text{if~} \Delta\geq \frac{2}{3} \,,\\[1ex]
            \frac{\Theta(1)}{\epsilon} \, 4^{\,f(\lambda_1,\,\Delta)}, & \text{if~} 0<\Delta<\frac{2}{3} \,.
        \end{cases}
    \end{equation}
Here $\Theta(1)$ is an absolute number; $f(\lambda_1,\Delta)$ is given by \cref{eq:defoff}.
\end{theorem}
\begin{proof}
We follow the same strategy as sec 3.3 in \cite{childs2025streaming}.
Lemma 10 therein still holds for general mixed states, which we quote here (and adapt it to our notation):
\begin{lemma}
    The expected sample complexity is
    \begin{equation}
        C(n)=\frac{2^n}{\prod_{k=0}^{n-1} p_k}.
    \end{equation}
\end{lemma}
\noindent Here, $n$ is the iteration level of the algorithm; $p_k$ is the probability for the SWAP test to have outcome $a=0$ at level $k$ 
(as described in \cref{sec:recursive}).  In particular, 
\begin{equation}
     p_k =\frac{1}{2}(1+\Tr(\rho_k^2)),
\end{equation}
where $\rho_k$ is the input state for level $k$.

\noindent$\bullet$ \textbf{Case 1:} $\Delta \geq 2/3$.
We have
\begin{equation}
    p_k \geq \frac{1+(\lambda_1^{(k)})^2}{2}\geq\lambda_1^{(k)}
    \geq \Delta^{(k)}
    = 1-\delta^{(k)}.
\end{equation}
Using eqs.~(94)-(96) in \cite{childs2025streaming}, we get:
\begin{equation}
    \prod_{k=0}^{n-1} p_k \geq \prod_{k=0}^{n-1} \lambda_1^{(k)}
    > 1-2\delta^{(0)}.
\end{equation}
Therefore,
\begin{equation}\label{eq:Cupbound1}
    C(n) =\frac{2^n}{\prod_{k=0}^{n-1} p_k}\leq \frac{2^n}{1-2\delta^{(0)}}
    =\frac{2^n}{2\Delta-1} =\frac{1-\lambda_1}{(2\Delta-1)^2\epsilon}.
\end{equation}
Here, $n$ is given by the first case of \cref{eq:Nupperbound}.
In particular, since $\Delta\geq \frac{2}{3}$, we have
\begin{equation}\label{eq:Cupbound11}
    C(n) \leq \frac{9(1-\lambda_1)}{\epsilon}.
\end{equation}

\noindent$\bullet$ \textbf{Case 2:} $\Delta<2/3$.
Since $p_k \geq \frac{1}{2}$, we replace $p_k$ by $\frac{1}{2}$ for all $k\geq n^*$ and get:
\begin{equation}
    C(n) \leq \frac{2^{n^*}}{\prod_{k=n-n*}^{n-1} p_k }\cdot 4^{n-n^*},
\end{equation}
where $n^*$ is the number of iterations needed to bring $\delta$ down to $\delta^{(n)}$ (\cref{eq:finaldelta}) from $\frac{1}{3}$, which is given by the logarithms in the second and third case in \cref{eq:Nupperbound}. 
The first factor in the above inequality follows the same estimation as in case 1, hence is upper bounded by \cref{eq:Cupbound11}.
Therefore, in this case,
\begin{equation}\label{eq:Cupbound2}
    C(n)\leq \frac{9(1-\lambda_1)}{\epsilon} \, 4^{\,5+f(\lambda_1,\,\Delta)}.
\end{equation}
\end{proof}

\subsection{Proof 2 -- New direct proof}\label{sec:proof2}
In our second proof, we present a new and more direct analysis of the same algorithm. 
In our analysis, we divide the iterations into three stages according to the progression of intermediate states' purity.
Starting from a highly noisy initial state, where the principal and second eigenvalues may be quite close, the algorithm progresses through the following three stages:
\begin{itemize}
    \item First stage: in a worst-case state, the gap $\Delta := \lambda_1 - \lambda_2$ can be arbitrarily small. Initially, progress is made by increasing $\Delta$ while there is little change in $\lambda_1$. For technical reason, progress is measured by a sum of eigenvalue ratios, and the first stage ends when  
    \begin{equation}
        \sum_{j=2}^{d} \frac{\lambda_j}{\lambda_1 - \lambda_j} \leq \frac{1}{\lambda_1^{(0)}},
    \end{equation}
    where $\lambda_1^{(0)}$ is the initial value of $\lambda_1$.\footnote{In principle we could use the contemporaneous value of $\lambda_1$ here. It makes no significant difference for worst-case (depolarized) states, but we use the initial value for simplicity of analysis later on.} 
    \item Second stage: at this point, the principal eigenvalue ($\lambda_1$) begins to grow and different ideas are required since $\lambda_1$ is changing. Surprisingly, the same sum can be used to measure progress, and the second stage ends when
    \begin{equation}
        \sum_{j=2}^{d} \frac{\lambda_j}{\lambda_1 - \lambda_j} \leq \frac{1}{2}.
    \end{equation}
    \item Third stage: the algorithm continues until the desired purification is achieved, i.e., $1-\lambda_1 \leq \epsilon$. We switch to working with $\eta := 1 - \lambda_1$ as a measure of progress.
\end{itemize}

We emphasize that the algorithm itself remains unchanged across all three stages.
The division into stages is solely for analytical convenience, enabling the use of different techniques and progress measures at each step.
Moreover, if the initial state is already somewhat pure, certain stages may be skipped.

\subsubsection{First Stage}

For the first two stages, we introduce the quantity 
\begin{equation}
\gamma_j := \gamma_j(\vlambda) = \frac{\lambda_j}{\lambda_1 - \lambda_j},
\end{equation}
which we most often use in a summation:
\begin{equation}
\label{eq:def_e}
E(\vlambda) := \sum_{j=2}^{d} \gamma_j(\vlambda) = \sum_{j=2}^{d} \frac{\lambda_j}{\lambda_1 - \lambda_j}.
\end{equation}
The quantity $E$ serves as an important measure of progress in the algorithm: it can be arbitrarily large initially and we will show that it decreases (see, e.g., \cref{cor:e_ratio_by_lambda} and \cref{lem:e_ratio}) from one level of the recursive swap test algorithm to the next. Recall that stage 1 runs until $E(\vlambda) \leq \tfrac{1}{\lambda_1}$. We will show that $E$ is bounded between
\begin{equation}
    \frac{1}{\lambda_1} - 1 \leq E(\vlambda) \leq \frac{1-\lambda_1}{\Delta}.
\end{equation}
The intuition for the threshold between stages is that we can bound the progress in $E$ as a function of $\frac{1}{\lambda_1}$, or use $\frac{1}{\lambda_1} \leq E + 1$. Towards the upper end of this interval, it is better to use $\frac{1}{\lambda_1}$ (\cref{cor:e_ratio_by_lambda}), but eventually we need to use $E$ (\cref{lem:e_ratio}) because $E$ itself is changing.

We begin with a lemma showing that $\gamma_j$ decreases geometrically, by a factor that is some function of $\lambda_1$. 
\begin{lemma}
	\label{lem:gamma_ratio}
	For all $\vlambda \in V$,
	\begin{equation}
	\frac{\gamma_j(\vlambda')}{\gamma_j(\vlambda)} = \frac{1 + \lambda_j}{1 + \lambda_1 + \lambda_j} \leq \exp\of*{-\frac{1}{1/\lambda_1 + \tfrac{3}{2}}},
	\end{equation}
	where $\vlambda' = F(\vlambda)$. 
\end{lemma}
\begin{proof}
	Recall from \cref{eq:evals} that for all $j$, 
	\begin{equation}
	\lambda_j' = \frac{\lambda_j(\lambda_j + 1)}{1 + \norm{\vlambda}^2}.
	\end{equation}
	In $\gamma_j(\vlambda')$, the $1 + \norm{\vlambda}^2$ denominators cancel out and we have 
	\begin{equation}
	\gamma_j(\vlambda') = \frac{\lambda_j + \lambda_j^2}{\lambda_1 + \lambda_1^2 - \lambda_j - \lambda_j^2} = \frac{\lambda_j}{\lambda_1 - \lambda_j} \frac{1 + \lambda_j}{1 + \lambda_1 + \lambda_j} = \gamma_j(\vlambda) \frac{1 + \lambda_j}{1 + \lambda_1 + \lambda_j}.
	\end{equation}
	A standard inequality on logarithms says $\frac{2}{2+x} \leq \frac{\nlog(1+x)}{x}$ for all $x > -1$. In particular, we have $\frac{2x}{2+x} \leq \nlog(1+x)$ for all $x > 0$. Substituting $x = \frac{\lambda_1}{1 + \lambda_j}$ gives 
	\begin{equation}
	\frac{\lambda_1}{1 + \lambda_j + \frac{1}{2}\lambda_1} = \frac{2x}{2+x} \leq \nlog(1 + x) = \nlog\of*{\frac{1 + \lambda_j + \lambda_1}{1 + \lambda_j}}.
	\end{equation}
	We rearrange this into
	\begin{equation}
	\frac{\gamma_j(\vlambda')}{\gamma_j(\vlambda)} = \frac{1 + \lambda_j}{1 + \lambda_1 + \lambda_j} \leq \exp\of*{-\frac{\lambda_1}{1 +  \lambda_j + \tfrac{1}{2} \lambda_1}} \leq \exp\of*{-\frac{1}{1/\lambda_1 + \tfrac{3}{2}}},
	\end{equation}
	where the rightmost inequality comes from setting $\lambda_j$ to its upper bound, $\lambda_1$. This completes the proof.
\end{proof}

While it is not strictly necessary for the proof, we also can bound $\gamma_j(\vlambda)$ on the other side by a similar function of $\lambda_1$. 
\begin{lemma}
	For all $\vlambda \in V$,
	\begin{equation}
	\exp\of*{-\lambda_1} \leq \frac{\gamma_j(\vlambda')}{\gamma_j(\vlambda)},
	\end{equation}
	where $\vlambda' = F(\vlambda)$. 
\end{lemma}
\begin{proof}
We have
    \begin{equation}
        \ln\frac{\gamma_j(\vlambda)}{\gamma_j(\vlambda')}
        =\ln\left( \frac{1 + \lambda_1 + \lambda_j}{1 + \lambda_j}   \right)
        < \frac{\lambda_1}{1+\lambda_j}
        \leq \lambda_1,
    \end{equation}
where in the first inequality we apply the basic inequality $\ln(1+x) \leq x$ for $x=\frac{\lambda_1}{1+\lambda_j}$.
The lemma is proved by exponentiating both sides and then taking their reciprocals.
\end{proof}

In other words, $\gamma_j(\vlambda)$ decreases by a factor of approximately $\exp(-\lambda_1)$ for all $j$. Since $\lambda_1$ is at least $\frac{1}{d}$, we can use this to lower bound the progress we make even when the gap is tiny. Since the same bounds hold for $\gamma_j(\vlambda') / \gamma_j(\vlambda)$ for all $j$, we can get a similar bound on $E(\vlambda') / E(\vlambda)$ quite easily.
\begin{corollary}
	\label{cor:e_ratio_by_lambda}
	For all $\vlambda \in V$,
	\begin{equation}
	\frac{E(\vlambda')}{E(\vlambda)} \leq \exp\of*{-\frac{1}{1/\lambda_1 + \tfrac{3}{2}}},
	\end{equation}
	where $\vlambda' = F(\vlambda)$. 
\end{corollary}
\begin{proof}
	\begin{align}
		\frac{E(\vlambda')}{E(\vlambda)} &= \frac{\sum_{j=2}^{d} \gamma_j(\vlambda')}{\sum_{j=2}^{d} \gamma_j(\vlambda)} && \text{by \eqref{eq:def_e},} \\
		&\leq \max_{j} \frac{\gamma_j(\vlambda')}{\gamma_j(\vlambda)} \\
		&\leq \exp\of*{-\frac{1}{1/\lambda_1 + \tfrac{3}{2}}} && \text{by \cref{lem:gamma_ratio}}. \qedhere
	\end{align}
\end{proof}

\begin{lemma}
    \label{lem:stage1_length}
	For any initial state $\vlambda \in V$, we have $E(\vlambda) \leq \tfrac{1}{\lambda_1}$ already or achieve $E(F^{(l_1)}(\vlambda)) \leq \frac{1}{\lambda_1}$ within $l_1 \leq \lceil (1/\lambda_1 + \tfrac{3}{2}) \nlog\of*{E(\vlambda) \lambda_1} \rceil$ levels\footnote{Note that this is the initial value of $\lambda_1$, \emph{not} $F^{(l_1)}(\vlambda)_1$, the value at step $l_1$.}.
\end{lemma}
\begin{proof}
	Take $E(F^{(t)}(\vlambda))/E(\vlambda)$, telescope it, and bound each term with \cref{cor:e_ratio_by_lambda}. 	
	\begin{equation}
		\frac{E(F^{(t)}(\vlambda))}{E(\vlambda)} = \prod_{k=0}^{t-1}\frac{E(F^{(k+1)}(\vlambda))}{E(F^{(k)}(\vlambda))} \leq \prod_{k=0}^{t-1} \exp\of*{-\frac{1}{1/(F^{(k)}(\vlambda))_1 + \tfrac{3}{2}}}
	\end{equation}
	Observe that $\exp(-\tfrac{1}{1/x + 3/2})$ is decreasing as a function of $x$ and since $F^{(k)}(\vlambda)_1$ is always at least $\lambda_1$, it follows that
	\begin{equation}
	\frac{E(F^{(t)}(\vlambda))}{E(\vlambda)} \leq \exp\of*{-\frac{t}{1/\lambda_1 + \frac{3}{2}}}.
	\end{equation}
	At $t = \ceiling*{(1/\lambda_1 + \tfrac{3}{2})\nlog\of*{E(\vlambda) \lambda_1}}$ we have 
	\begin{equation}
	E(F^{(l_1)}(\vlambda)) \leq E(\vlambda) \exp\of*{-\frac{\ceiling*{(1/\lambda_1 + \tfrac{3}{2})\nlog\of*{E(\vlambda) \lambda_1}}}{1/\lambda_1 + \tfrac{3}{2}}} \leq \frac{1}{\lambda_1}, 
	\end{equation}
    and since $l_1$ is the first index where this happens, we have our desired bound on $l_1$.
\end{proof}

\subsubsection{Second Stage}

The second stage begins when $E \leq \frac{1}{\lambda_1}$. We will show in a moment that $\frac{1}{\lambda_1} - 1 \leq E$, so $E$ and $\tfrac{1}{\lambda_1}$ are approximately the same. Unfortunately, this means that progress in $E$ changes the rate of progress (tied to $\tfrac{1}{\lambda_1}$). In this section we introduce new bounds to handle the dynamic nature of this situation.

\begin{proposition}
	\label{prop:e_bounds_lambda1}
	\( \frac{1}{\lambda_1} \leq 1 + E(\vlambda) \) for all $\vlambda \in V$.	
\end{proposition}
\begin{proof}
	The proof is a one-liner:
	\begin{equation*}
	\frac{1}{\lambda_1} = \sum_{j=1}^{d} \frac{\lambda_j}{\lambda_1} \leq 1 + \sum_{j=2}^{d} \frac{\lambda_j}{\lambda_1 - \lambda_j} = 1 + \sum_{j=2}^{d} \gamma_j(\vlambda) = 1 + E(\vlambda). \qedhere
	\end{equation*}
\end{proof} 
\noindent Conveniently, \cref{prop:e_bounds_lambda1} combines neatly with \cref{lem:gamma_ratio} to bound the rate of decay for $\gamma_j$, and by extension, for $E$.  
\begin{lemma}
	\label{lem:e_ratio}
	For all $\vlambda \in V$ and $\lambda' = F(\vlambda)$, 
	\begin{equation}
	\frac{E(\vlambda')}{E(\vlambda)} \leq \exp\of*{-\frac{1}{E(\vlambda) + \tfrac{5}{2}}}.
	\end{equation}
\end{lemma}
\begin{proof}
    Bound the ratio with \cref{cor:e_ratio_by_lambda} and then use \cref{prop:e_bounds_lambda1} inside the exponential.
\end{proof}
\noindent
\cref{lem:e_ratio} says that the rate of decay of $E$ is bounded by a function in $E$ itself. For the sake of simpler analysis, we bound the number of levels for $E$ to decrease by a constant factor $\alpha$. 
\begin{corollary}
	\label{cor:cut_e}
	Fix $\alpha \in (0,1)$. For all $\vlambda \in V$, $E(F^{(t)}(\vlambda)) \leq \alpha E(\vlambda)$ for some $t \in \mathbb N$ at most 
    \begin{equation}
	t \leq \ceiling*{\of*{E(\vlambda) + \tfrac{5}{2}} \nlog(\tfrac{1}{\alpha})}.
	\end{equation}
\end{corollary}
\begin{proof}
	We telescope (much like the proof of \cref{lem:stage1_length}), use \cref{lem:e_ratio} on each factor, and use the last term to upper bound the rest since $E$ is clearly decreasing with each iteration.
	\begin{equation}
		\frac{E(F^{(t)}(\vlambda))}{E(\vlambda)} = \prod_{k=0}^{t-1}\frac{E(F^{(k+1)}(\vlambda))}{E(F^{(k)}(\vlambda))} \leq \prod_{k=0}^{t-1} \exp\of*{-\frac{1}{E(F^{(k)}(\vlambda)) + \tfrac{5}{2}}} \leq \exp\of*{-\frac{t}{E(\vlambda) + \tfrac{5}{2}}} \,.
	\end{equation}
	We can make the RHS less than $\alpha$ with any $t \geq   \of*{E(\vlambda) + \tfrac{5}{2}} \nlog(\tfrac{1}{\alpha})$, so we take the ceiling.
	\begin{equation}
	t \leq \ceiling*{\of*{E(\vlambda) + \tfrac{5}{2}} \nlog(\tfrac{1}{\alpha})}. \qedhere
	\end{equation}
\end{proof}

\begin{lemma}
	\label{lem:stage2_length}
	For $\vlambda \in V$ such that $\gamma_2(\vlambda) \leq 1$, the minimum $l_2 \in \mathbb N$ for which $E(F^{(l_2)}(\vlambda)) \leq \tfrac{1}{2}$ satisfies
	\begin{equation}
	l_2 \leq 2E(\vlambda)\nlog 2 + (1 + \tfrac{5}{2} \nlog 2) \ceiling*{1 + \frac{\nlog E(\vlambda)}{\nlog 2}}.
	\end{equation}
\end{lemma}
\begin{proof}
    Let us use \cref{cor:cut_e} with $\alpha = 2$ to repeatedly cut $E = E(\vlambda)$ in half. This gives the recurrence 
	\begin{equation}
	\st(E) = 
	\begin{cases}
		\st(E/2) + (\nlog 2)(E + \tfrac{5}{2}) + 1, &\text{if $E > \tfrac{1}{2}$,} \\
		0, & \text{if $E \leq \tfrac{1}{2}$.}		
	\end{cases}
	\end{equation}
	where $\st(E)$ bounds the number of steps to reduce $E$ below $\tfrac{1}{2}$. We claim the solution to this recurrence satisfies 
	\begin{equation}
	\st(E) \leq 2E \nlog 2 + (1 + \tfrac{5}{2} \nlog 2) \ceiling*{1 + \frac{\nlog E}{\nlog 2}},
	\end{equation}
	as long as $E \geq \tfrac{1}{2}$. This can be confirmed with a simple induction: 
	\begin{align}
		\st(E) &= \st(E/2) + (\nlog 2)(E + \tfrac{5}{2}) + 1 \\
		&\leq 2(E/2) \nlog 2 + (1 + \tfrac{5}{2} \nlog 2) \ceiling*{\frac{\nlog E}{\nlog 2}}  + (\nlog 2)(E + \tfrac{5}{2}) + 1 \\
		&= 2E \nlog 2+ (1 + \tfrac{5}{2} \nlog 2) \ceiling*{\frac{\nlog E}{\nlog 2} + 1}.
	\end{align} 
\end{proof}

\subsubsection{Third Stage}

In the final stage, we boost $\lambda_1$ up to our desired level of fidelity, $1- \epsilon$. We use the progress measure $\eta(\vlambda) := 1 - \lambda_1$, and we show it approaches zero geometrically. It will be important to track the probability of successful swap tests because for the first time in this algorithm, it can be significantly higher than $\tfrac{1}{2}$.

\begin{lemma}
    \label{lem:eta_leq_e}
    For all $\vlambda \in V$, $\eta(\vlambda) \leq E(\vlambda)$.
\end{lemma}
\begin{proof}
	\cref{prop:e_bounds_lambda1} gives \( \frac{1}{\lambda_1} \leq 1 + E(\vlambda) \), which we rearrange to 
	\begin{equation}
	\eta(\vlambda) = 1 - \lambda_1 \leq \frac{1 - \lambda_1}{\lambda_1} \leq E(\vlambda). \qedhere
	\end{equation}
\end{proof}

As in the second regime, $E$ decreases geometrically (\cref{lem:e_ratio}). However, this bound scales poorly as $E \to 0$. Instead, we prove a new bound in terms of $\eta$ which gives the best ratio ($\tfrac{1}{2}$) as $\eta \to 0$. 
\begin{lemma}
    \label{lem:e_ratio_stage3}
    For all $\vlambda \in V$ and $\vlambda' = F(\vlambda)$, 
    \begin{equation}
    \frac{E(\vlambda')}{E(\vlambda)} \leq \frac{1+\eta(\vlambda)}{2} \leq \frac{1}{2} \exp(\eta(\vlambda)).
    \end{equation}
\end{lemma}    
\begin{proof}
The proof is a simple calculation and shares some similarity with the proof of \cref{lem:e_ratio}.
\begin{align}
	\frac{E(\vlambda')}{E(\vlambda)} &\leq \max_{j} \frac{\gamma_j(\vlambda')}{\gamma_j(\vlambda)} \\
	&= \max_{j} \frac{1 + \lambda_j}{1 + \lambda_1 + \lambda_j} && \text{by \cref{lem:gamma_ratio}} \\
	&\leq \frac{1 + \eta(\vlambda)}{1 + \lambda_1 + \eta(\vlambda)} && \text{since $\lambda_j \leq \eta(\vlambda)$ for $j\geq 2$} \\
	&= \frac{1 + \eta(\vlambda)}{2}
\end{align}
We finally bound this with $\tfrac{1}{2} \exp(\eta(\vlambda))$ using the standard inequality $1+x \leq \exp(x)$.
\end{proof}

\begin{lemma}
    \label{lem:stage3_steps}
    Let $\vlambda \in V$ be such that $E(\vlambda) \leq \tfrac{1}{2}$. Then $\eta(F^{(t)}(\vlambda)) \leq \frac{e^2}{2^{t}} E(\vlambda)$ for all $t \in \mathbb N$. 
\end{lemma}
\begin{proof}
	We begin with $E(\vlambda) \leq \tfrac{1}{2}$, and it is clear from \cref{lem:e_ratio} or \cref{lem:e_ratio_stage3} that $E$ is decreasing over time. In fact, \cref{lem:e_ratio_stage3} shows it decreases geometrically with ratio $\leq \tfrac{3}{4}$. By an easy induction, 
	\begin{equation}
		\label{eq:eta_decay}
		\eta(F^{(t)}(\vlambda)) \leq E(F^{(t)}(\vlambda)) \leq E(\vlambda) \of*{\frac{3}{4}}^{t} \leq \frac{1}{2} \of*{\frac{3}{4}}^{t}.
	\end{equation}
    
    Finally, \cref{lem:e_ratio_stage3} also gives a bound based on an exponential.
    \begin{align}
    E(F^{(t)}(\vlambda)) &\leq E(\vlambda) \prod_{k=0}^{t-1} \tfrac{1}{2} \exp(\eta(F^{(k)}(\vlambda)) && \text{by telescoping and \cref{lem:e_ratio_stage3}} \\
    &= \frac{E(\vlambda)}{2^t} \exp \of*{\sum_{k=0}^{t-1} \eta(F^{(k)}(\vlambda))} \\
    &\leq \frac{E(\vlambda)}{2^{t}} \exp \of*{\frac{1}{2} \sum_{k=0}^{\infty} \of*{\frac{3}{4}}^{k}} && \text{using \eqref{eq:eta_decay}} \\
    &= \frac{e^2}{2^{t}} E(\vlambda).
    \end{align}
    Using \cref{lem:eta_leq_e} ($\eta(F^{(t)}(\vlambda)) \leq E(F^{(t)}(\vlambda))$) again finishes the proof.
\end{proof}

Finally, let $c(\vlambda)$ be the expected number of samples to achieve one successful swap test (and thus generate one sample for the next level). 
\begin{proposition}
    \label{prop:cost}
	For $\vlambda \in V$, 
	\begin{equation}
	c(\vlambda) \leq 2 \exp(\eta(\vlambda)).
	\end{equation}
\end{proposition}
\begin{proof}
	The cost per swap test is $2$ samples, and the probability of success is \( p(\vlambda) = \frac{1 + \norm{\vlambda}^2}{2} \), so the expected cost is \( c(\vlambda) = \frac{4}{1 + \norm{\vlambda}^2} \). We bound this as follows:
	\begin{equation}
	\frac{4}{1 + \norm{\vlambda}^2} \leq \frac{4}{1 + \lambda_1^2} = \frac{2}{1 - \eta(\vlambda) + \tfrac{1}{2}\eta(\vlambda)^2}
	\end{equation}
    It is a standard result in calculus that the difference between a smooth function ($\exp(-x)$, in this case) and its truncated Taylor series (i.e., $1 - x + \tfrac{1}{2}x^2$),  
    \begin{equation}
    (1 - x + \tfrac{1}{2}x^2) - \exp(-x) = \tfrac{1}{6}\xi^3 \geq 0,
    \end{equation}
    is the next term of the series evaluated at some $\xi \in [0,x]$. Thus, $\frac{1}{1-x+\tfrac{1}{2}x^2} \leq \exp(x)$ for $x \geq 0$, and the result follows.
\end{proof}

\begin{lemma}
	\label{lem:stage3_cost}
	Let $\vlambda \in V$ be such that $E(\vlambda) \leq \tfrac{1}{2}$. Given $\epsilon > 0$, the expected cost to produce $1$ purified sample is at most $2e^4 E(\vlambda) / \epsilon$.
\end{lemma}
\begin{proof}
	\cref{lem:stage3_steps} gives an upper bound on the number of steps to reach a desired purity.
    Specifically, if we want $\eta(F^{(t)}(\vlambda)) \leq \frac{e^2}{2^{t}} E(\vlambda) \leq \epsilon$ then we may choose the minimal $t$ such that $2^{t} \geq \tfrac{1}{\epsilon} e^2 E(\vlambda)$. Since $t$ is an integer, we may overshoot, but the above chosen $t$ must satisfy: 
    \begin{equation}
        \label{eq:stage3_cost_2t}
        2^{t} \leq \frac{2e^{2} E(\vlambda)}{\epsilon}.
    \end{equation}
	Meanwhile, the total cost over $t$ steps is 
	\begin{align}
		\prod_{k=0}^{t-1} c(F^{(k)}(\vlambda)) &\leq \prod_{k=0}^{t-1} 2 \exp(\eta(F^{(k)}(\vlambda))) && \text{by \cref{prop:cost}} \\
		&\leq 2^{t} \exp\of*{\sum_{k=0}^{\infty} \eta(F^{(k)}(\vlambda))} \\
        &\leq 2^{t} \exp \of*{\frac{1}{2} \sum_{k=0}^{\infty} \of*{\frac{3}{4}}^{k}} && \text{using \eqref{eq:eta_decay}} \\
		&\leq 2^{t} e^2 \\
        &\leq \frac{2e^4 E(\vlambda)}{\epsilon} && \text{by \eqref{eq:stage3_cost_2t}}
	\end{align}
	by a similar argument to \cref{lem:stage3_steps}.
\end{proof}

\subsubsection{Sample Complexity Upper Bound}

We have an analysis for each of the three stages with \cref{lem:stage1_length}, \cref{lem:stage2_length}, and \cref{lem:stage3_cost}. We now combine them to give an expected sample complexity across various parameter regimes. 
\begin{theorem}
\label{thm:proof2_sample_complexity}
    Given $\vlambda \in V$ and $\epsilon > 0$, the expected sample complexity of one purified sample of fidelity $\geq 1 - \epsilon$ is at most 
    \begin{equation}
    \label{eq:proof-2-complexity}
    \begin{cases}
            \Theta \of*{ \frac{1-\lambda_1}{\epsilon}} & \text{if $\lambda_1 \geq \tfrac{2}{3}$,} \\
            \Theta \of*{ \frac{1}{\epsilon}} \cdot 4^{2 \ln 2\tfrac{1 - \lambda_1}{\Delta} + (\tfrac{5}{2} + \tfrac{1}{\ln 2}) \ln \frac{1-\lambda_1}{\Delta}} & \text{if $\lambda_1 \leq \tfrac{2}{3}$ and $\frac{1-\lambda_1}{\Delta} \leq \tfrac{1}{\lambda_1}$} \\
            \Theta \of*{ \frac{1}{\epsilon}}\cdot 4^{ \of*{\frac{1}{\lambda_1} - 1} \left[ \nlog\of*{\frac{(1-\lambda_1)\lambda_1}{\Delta}} + 2 \ln 2 \right] + \frac{5}{2}  \nlog\of*{\frac{1-\lambda_1}{\Delta}} + \log_2 \frac{1}{\lambda_1}} & \text{if $\frac{1-\lambda_1}{\Delta} \geq \tfrac{1}{\lambda_1}$} 
        \end{cases}
    \end{equation}
\end{theorem}
\begin{proof}
    The high-level proof idea is clear: the sample complexity is a product $C_1 C_2 C_3$ of the complexities $C_1$, $C_2$, $C_3$ for each stage. Moreover, we have already bounded each stage in some way. 
    \begin{enumerate}
        \item \Cref{lem:stage1_length} bounds the length of stage 1, 
        	\begin{equation}
            \label{eq:steps-stage-1}
            l_1 \leq (\frac{1}{\lambda_1} + \frac{3}{2}) \nlog\of*{E(\vlambda) \lambda_1} + O(1).
        	\end{equation}
        \item \Cref{lem:stage2_length} bounds the length of stage 2,
        	\begin{equation}
            \label{eq:steps-stage-2}
        	l_2 \leq 2E\nlog 2 + \of*{\frac{1}{\nlog 2} + \frac{5}{2}} \nlog E + O(1).
        	\end{equation}
        \item \Cref{lem:stage3_cost} bounds the sample complexity of stage 3,
            \begin{equation}
                \label{eq:complexity-stage-3}
                C_3 \leq O\of*{\frac{E(\vlambda)}{\epsilon}}.
            \end{equation}
    \end{enumerate}
    The probability of a successful swap test is always at least $\tfrac{1}{2}$. In stage 1 and stage 2, this is close to the truth, so we bound their sample complexities under the assumption that we use an expected factor of $4$ copies per level . I.e., $C_1 \leq 4^{l_1}$ and $C_2 \leq 4^{l_2}$. 

    However, depending on the initial and final purity of the state, the algorithm may not visit all three stages. To eliminate some of these cases, let us assume that $\epsilon \leq \tfrac{1}{3}$ so that the algorithm always ends in the third stage. If a lower fidelity is acceptable then our bounds will be loose but correct. 

    Now observe that we have the bound 
    \begin{equation}
        \label{eq:e-upper-substitute}
        E(\vlambda) = \sum_{j=2}^{d} \frac{\lambda_j}{\lambda_1 - \lambda_j} \leq \sum_{j=2}^{d} \frac{\lambda_j}{\Delta} = \frac{1-\lambda_1}{\Delta}.
    \end{equation}
    If the algorithm only spends time in stage 3, then we get a nice bound by using \cref{eq:e-upper-substitute} directly in \cref{eq:complexity-stage-3}: 
    \begin{equation}
        \label{eq:c3-stage-3}
        C_3 = O\of*{ \frac{1-\lambda_1}{\Delta \epsilon} } = O\of*{ \frac{1-\lambda_1}{\epsilon}}.
    \end{equation}
    Note that we can absorb $\Delta$ into the big-$\Theta$ here because in stage 3, $\tfrac{1}{2} \geq E(\vlambda) \geq \tfrac{1}{\lambda_1} - 1$ implies $\lambda_1 \geq \tfrac{2}{3}$. This means $\lambda_2$ is at most $\tfrac{1}{3}$ and thus $\Delta \geq \tfrac{1}{3}$ or $\tfrac{1}{\Delta} \leq 3$. 

    Now suppose the algorithm also spends some time in the second stage too. We cannot reuse \cref{eq:c3-stage-3}  because whatever $E(\vlambda)$ may be initially, stage 3 only begins when it has reduced to $\leq \tfrac{1}{2}$. With $E \leq \tfrac{1}{2}$ we get the following bound on stage 3. 
    \begin{equation}
    \label{eq:c23-stage-3}
        C_3 = O\of*{\frac{1}{\epsilon}}.
    \end{equation}
    For stage 2, we apply \cref{eq:e-upper-substitute} in \cref{eq:steps-stage-2}:
    \begin{equation}
        \label{eq:c23-stage-2}
        l_2 \leq 2 \ln 2 \frac{1 - \lambda_1}{\Delta} + (\frac{5}{2} + \frac{1}{\ln 2}) \ln \frac{1-\lambda_1}{\Delta} + O(1).
    \end{equation}
    
    Finally, suppose the algorithm passes through all three stages. In the first stage, \cref{eq:e-upper-substitute} and \cref{eq:steps-stage-1} give
    \begin{equation}
    \label{eq:c123-stage-1}
        l_1 \leq \of*{\frac{1}{\lambda_1} + \frac{3}{2}} \nlog\of*{\frac{(1-\lambda_1)\lambda_1}{\Delta}} + O(1).
    \end{equation}
    Going into phase 2, we may assume $E \leq \tfrac{1}{\lambda_1}$, which we combine with \cref{eq:steps-stage-2} to get
    \begin{equation}
        \label{eq:c123-stage-2}
        l_2 \leq \frac{2\nlog 2}{\lambda_1} + \of*{\frac{1}{\nlog 2} + \frac{5}{2}} \nlog \frac{1}{\lambda_1} + O(1).  
    \end{equation}
    We can then combine \cref{eq:c123-stage-1} with \cref{eq:c123-stage-2} and simplify. 
    \begin{align}
        l_1 + l_2 &\leq \of*{\frac{1}{\lambda_1} + \frac{3}{2}} \nlog\of*{\frac{(1-\lambda_1)\lambda_1}{\Delta}} + \frac{2\nlog 2}{\lambda_1} + \of*{\frac{1}{\nlog 2} + \frac{5}{2}} \nlog \frac{1}{\lambda_1} + O(1) \\
        &= \of*{\frac{1}{\lambda_1} - 1} \nlog\of*{\frac{(1-\lambda_1)\lambda_1}{\Delta}} + \frac{5}{2}  \nlog\of*{\frac{1-\lambda_1}{\Delta}} + \frac{2\nlog 2}{\lambda_1} + \frac{1}{\nlog 2} \nlog \frac{1}{\lambda_1} + O(1) \\
        &\leq \of*{\frac{1}{\lambda_1} - 1} \left[ \nlog\of*{\frac{(1-\lambda_1)\lambda_1}{\Delta}} + 2 \ln 2 \right] + \frac{5}{2}  \nlog\of*{\frac{1-\lambda_1}{\Delta}} + \log_2 \frac{1}{\lambda_1} + O(1). \label{eq:c123-stage-1-2}
    \end{align}
    Note that the complexity for stage 3 is the same as \cref{eq:c23-stage-3}.
    
    The final complexities arise from the following combinations of equations:
    \begin{align}
        & \,\begin{cases}
            \eqref{eq:c3-stage-3} & \text{if $\lambda_1 \geq \tfrac{2}{3}$,} \\
            \eqref{eq:c23-stage-3} \cdot 4^{\eqref{eq:c23-stage-2}} & \text{if $\lambda_1 \leq \tfrac{2}{3}$ and $\frac{1-\lambda_1}{\Delta} \leq \tfrac{1}{\lambda_1}$} \\
            \eqref{eq:c23-stage-3} \cdot 4^{\eqref{eq:c123-stage-1-2}} & \text{if $\frac{1-\lambda_1}{\Delta} \geq \tfrac{1}{\lambda_1}$} 
        \end{cases} \\
        =&
        \,\begin{cases}
            O \of*{ \frac{1-\lambda_1}{\epsilon}} & \text{if $\lambda_1 \geq \tfrac{2}{3}$,} \\
            O \of*{ \frac{1}{\epsilon}} \cdot 4^{2 \ln 2\tfrac{1 - \lambda_1}{\Delta} + (\tfrac{5}{2} + \tfrac{1}{\ln 2}) \ln \frac{1-\lambda_1}{\Delta}} & \text{if $\lambda_1 \leq \tfrac{2}{3}$ and $\frac{1-\lambda_1}{\Delta} \leq \tfrac{1}{\lambda_1}$} \\
            O \of*{ \frac{1}{\epsilon}}\cdot 4^{ \of*{\frac{1}{\lambda_1} - 1} \left[ \nlog\of*{\frac{(1-\lambda_1)\lambda_1}{\Delta}} + 2 \ln 2 \right] + \frac{5}{2}  \nlog\of*{\frac{1-\lambda_1}{\Delta}} + \log_2 \frac{1}{\lambda_1}} & \text{if $\frac{1-\lambda_1}{\Delta} \geq \tfrac{1}{\lambda_1}$} 
        \end{cases}
    \end{align}
\end{proof}

\section{Sample Complexity Lower bound}
\label{sec:lowbdd}

In this section, we derive three lower bounds (collected into two theorems) for the sample complexity to purify general noisy states.  
In each case, very few parameters about the noisy states are known, and the lower bound holds for the worse case spectra satisfying the constraints.    
Our derivations follow an approach inspired by, and similar to, that described in Section 5 of \cite{childs2025streaming}.  In each case, the general purification protocol is applied to solve a purification problem with known optimal sample complexity.  

The first two bounds are built upon an optimality result from 
\cite{CEM99}.  It is shown that, to perform the purification for a single qubit,
\begin{equation}
    \label{eq:cem99}
    \left( (1-\delta)\ket{\psi}\bra{\psi}+\delta\frac{I}{2} \right)^{\!\!\otimes N}\to
    ~\sigma
\end{equation}
such that $F(\sigma,\ket{\psi}\bra{\psi})\geq 1-\epsilon$,
the sample complexity is lower bounded by
\begin{equation}
        N\geq \frac{\delta}{2(1-\delta)^2\epsilon}-o_\delta \left(\frac{1}{\epsilon}\right).
        \label{eq:samcomcem}
\end{equation}
Here $o_\delta(\frac{1}{\epsilon})$ denotes a term that grows asymptotically slower than $\frac{1}{\epsilon}$ for any fixed $\delta$. 

From the above, we obtain the following sample complexity lower bounds for general purification.

\begin{theorem} ~~~\\[-3ex]
    \begin{enumerate}
        \item To purify a general mixed state $\rho$ with dimension $d$ and $\lambda_1 > \tfrac{1}{2}$, so that the fidelity to its principal eigenvector becomes at least $1-\epsilon$, the sample complexity in the worst case is at least
    \begin{equation}
    N\geq\frac{1-\lambda_1}{(2 \lambda_1-1)^2 \epsilon}  -  o_{\lambda_1}\left(\frac{1}{\epsilon}\right).
    \end{equation}
    \item To purify a general mixed state $\rho$ with dimension $d$ and gap $\Delta$, so that the fidelity to its principal eigenvector becomes at least $1-\epsilon$, the sample complexity in the worst case is at least
    \begin{equation}
    N\geq\frac{1-\Delta}{2 \Delta ^2 \epsilon}-o_\Delta \left(\frac{1}{\epsilon}\right).
    \end{equation}
    \end{enumerate}
\end{theorem}
The lower bounds from the above theorem can be compared with the upper bounds derived for the recursive swap test purification protocol 
\cref{eq:Cupbound1,eq:Cupbound2}, or  \cref{eq:proof-2-complexity}.  
We can see that if $\lambda_1$ is bounded away from $1/2$ (for example, $\lambda_1\geq 2/3$), then our protocol is optimal up to an absolute constant factor in the sense that it works for general mixed state with dimension $d$ and principal eigenvalue $\lambda_1$ and matches the lower bound in the worst case.
Similarly, if $\Delta$ is bounded away from 0 (and the distance is independent of $d$, for example, $\Delta\geq 1/3$), then our protocol is optimal up to an absolute constant factor in the sense that it works for general mixed state with dimension $d$ and gap $\Delta$ and matches the lower bound in the worst case.

\begin{proof}
    Given a general state purification protocol, we can design a qubit purification protocol as follows:
    \begin{equation}
        \rho_2^{\otimes N} \xrightarrow{\text{embedding}} \rho_d^{\otimes N} \xrightarrow{\text{purification}} \sigma_d
        \xrightarrow{\text{reduction}} \sigma_2,
    \end{equation}
    introducing quantum channels which embed a qubit into a $d$-dimensional qudit, and extract a qubit from a qudit. 
    
    More precisely, we suppose that the qubit is the subspace of the qudit spanned by the first two basis vectors, $\ket{1}$ and $\ket{2}$. Then an arbitrary single-qubit mixed state $\rho_2 = (1-\delta)\ket{\psi}\bra{\psi}+\delta\frac{I}{2}$ can be embedded into the qudit by inclusion:  
    \begin{equation}
        \rho_2 \mapsto \rho_d=\rho_2+0\sum_{i=3}^{d}\ket{i}\bra{i}.
    \end{equation}
    Here,
    \begin{equation}\label{eq:embedDelta}
        \lambda_1(\rho_d)=1-\frac{\delta}{2}, ~~\Delta(\rho_d)=1-\delta. 
    \end{equation}
    Then we perform the purification protocol $\mathcal{P}$ for $\rho_d^{\otimes N}$:  
    \begin{equation}
        \rho_d^{\otimes N} \mapsto \sigma_d=\mathcal{P}(\rho_d^{\otimes N}).
    \end{equation}
    Then we reduce $\sigma_d$ back to a qubit state as follows:
    \begin{equation}
        \sigma_d \mapsto \sigma_2=
        P\sigma_d P+\Tr((1-P)\sigma_d)\frac{I}{2},
    \end{equation}
    where $P$ is the projector onto the subspace spanned by $\ket{1}$ and $\ket{2}$. That is, we make the two-outcome measurement with projectors $P$ and $I - P$, and replace the state with a maximally mixed qubit for the $I - P$ outcome (and keep the $P$ outcome state). 

    Note that 
    \begin{equation}
    \bra{\psi}\sigma_2\ket{\psi}
    =\bra{\psi}\sigma_d\ket{\psi}+\Tr((1-P)\sigma_d) \bra{\psi}\tfrac{I}{2}\ket{\psi}
    \geq \bra{\psi}\sigma_d\ket{\psi}.
    \end{equation}
    Therefore, to purify $\rho_2$ so that $\bra{\psi}\sigma_2\ket{\psi}\geq 1-\epsilon$, it suffices to purify $\rho_d$ so that $\bra{\psi}\sigma_d\ket{\psi}\geq 1-\epsilon$.
    Therefore, to purify $\rho_d$ to $\epsilon$ infidelity, we also need
    the sample complexity according to \cref{eq:samcomcem}:
    \begin{equation}
        N\geq\frac{\delta}{2(1-\delta)^2\epsilon}-o_\delta
        \left(\frac{1}{\epsilon}\right).
    \end{equation}
    Substituting $\e$ with $\e = 2(1-\lambda_1(\rho_d))$ and with $\e = 1-\Delta(\rho_d)$ (immediate from \cref{eq:embedDelta}) gives statements 1 and 2 of the theorem respectively.      
\end{proof}

Note that \cref{eq:embedDelta} requires $\lambda_1>1/2$ for allowing a solution with $\delta<1$.  Furthermore, while our proof is similar to that employed in Section 5 of \cite{childs2025streaming}, applying it here to the more general input state significantly simplifies the proof, and strengthens the bound to be dimension independent.   
In the current derivation, the state after the embedding need not be a depolarized state, which enables us to tune the embedding parameter more favorably.

\vspace*{2ex}

Our third lower bound is built on the more recent optimality result from 
\cite{LFIC24}, wherein \cref{eq:samcomcem} is generalized to qudits.
It is shown that to perform
\begin{equation}
    \of*{(1-\delta)\ket{\psi}\bra{\psi}+\delta\frac{I}{d}}^{\!\!\otimes N}\to 
    ~\sigma
\end{equation}
such that $F(\sigma,\ket{\psi}\bra{\psi})\geq 1-\epsilon$,
the sample complexity is \emph{asymptotically} lower bounded by
\begin{equation}\label{eq:NlowerboundquDit}
        N\geq \left(1-\frac{1}{d}\right)\frac{\delta}{(1-\delta)^2\epsilon}-o_{\delta,d}\left(\frac{1}{\epsilon}\right).
\end{equation}
Here $o_{\delta,d}(\frac{1}{\epsilon})$ denotes a term that grows asymptotically slower than $\frac{1}{\epsilon}$ for any fixed $\delta$ and $d$.

With this in mind, we can prove the following result.
\begin{theorem}
    To purify a general mixed state with dimension $d$, principal and second principal eigenvalues $\lambda_1$ and $\lambda_2$, so that the fidelity to its principal eigenvector becomes at least $1-\epsilon$, the sample complexity in the worst case is at least
    \begin{equation}
        N > \frac{1-\lambda_1 }{  4\Delta^2\epsilon    }-o_{\lambda_1,\lambda_2}\left(\frac{1}{\epsilon}\right).
    \end{equation}
\end{theorem}
Compare the above theorem with \cref{eq:Nupperboundthm1,eq:proof-2-complexity}, we see that if $\Delta$ is bounded away from 0 (and the distance is independent of $d$), then our protocol is optimal up to an absolute constant factor in the sense that it works for general mixed state with given dimension $d$, principal and second principal eigenvalues $\lambda_1$ and $\lambda_2$ and matches the lower bound in the worst case.

\begin{proof}
    We use a similar proof strategy as before, except that now we embed $d'$-dimensional qudits, with $d'$ to be determined later. More precisely, given $\rho_{d'}=(1-\delta)\ket{\psi}\bra{\psi}+\delta \frac{I}{d'}$, we embed it into a $d$-dimensional space as follows:
    \begin{equation}\label{eq:embed2}
        \rho_{d'} \mapsto \rho_d=\mu\rho_{d'}+(1-\mu)\frac{I_{d-d'}}{d-d'},
    \end{equation}
    where $0 \leq \mu \leq 1$ and $I_{d-d'}/(d-d')$ is the maximally mixed state in the space orthogonal to $\rho_{d'}$.
    For this state, we have
    \begin{equation}
        \lambda_1(\rho_d)=\mu(1-\frac{d'-1}{d'}\delta), ~~\lambda_2(\rho_d)=\frac{\mu\delta}{d'},
    \end{equation}
    assuming 
    \begin{equation}\label{eq:condition1}
        \mu\delta/d' \geq (1-\mu)/(d-d').
    \end{equation}
    Solving $\mu$ and $\delta$, we get
    \begin{equation}\label{eq:solvedelta2}
        \delta=\frac{d'\lambda_2}{\lambda_1+(d'-1)\lambda_2},~~
        \mu=\lambda_1+(d'-1)\lambda_2.
    \end{equation}
    The requirement \cref{eq:condition1} is then automatically satisfied, since it is equivalent to $\lambda_1+(d-1)\lambda_2\geq 1$. 
    Also, we need $0\leq \delta \leq 1$, and $0\leq \mu \leq 1$.
    The former is automatically satisfied since $\lambda_1>\lambda_2$, and the latter is achieved by choosing 
    \begin{equation}\label{eq:defdprime}
        d'=\floor*{\frac{1-\lambda_1}{\lambda_2}}+1.
    \end{equation}
    The above choice for $d'$ also implies $d'\leq d$, hence this is a valid choice.
    To complete the analysis, we substitute 
    \cref{eq:solvedelta2,eq:defdprime} into \cref{eq:NlowerboundquDit} and note that 
    \begin{equation}
        \left(1-\frac{1}{d'}\right)\frac{\delta}{(1-\delta)^2\epsilon}=
        \of*{1-\frac{1}{d'}}\frac{
        d'\lambda_2(\lambda_1+(d'-1)\lambda_2)
        }{\Delta^2\epsilon}
        >\left(1-\frac{1}{d'} \right) \frac{     (1-\lambda_1)(1-\lambda_2)    }{ \Delta^2\epsilon }
        \geq \frac{ 1-\lambda_1 }{ 4\Delta^2\epsilon }.
    \end{equation}
\end{proof}

\section*{Acknowledgements}
We thank Andrew Childs, Honghao Fu, and Maris Ozols for permission to adapt figures, algorithm specifications, and some text and equations from \cite{childs2025streaming}, and Zhaoyi Li and Honghao Fu for sharing their recent results from \cite{LFIC24}.

DL received support from NSERC discovery grant and NSERC Alliance grant under the project QUORUM.
ZL is supported by Applied Quantum Computing Challenge Program from National Research Council of Canada.
DL and ZL, via the Perimeter Institute, are supported in part by the Government of Canada through the
Department of Innovation, Science and Economic Development and by the
Province of Ontario through the Ministry of Colleges and Universities.

\bibliographystyle{alphaurl}
\bibliography{purification}

\end{document}